\documentclass[twoside]{aiml18}

\usepackage{aiml18macro}

\usepackage{enumitem}
\setenumerate{itemsep=0pt}

\usepackage{amssymb}
\usepackage{MnSymbol}
\usepackage{hyperref}
\usepackage{tikz}
\usepackage{xspace}
\usepackage{multicol}
\hypersetup{bookmarks = true, colorlinks = false}

\newcommand{\LTC}{\ensuremath{\mathsf{LTC}}\xspace}

\newcommand{\X}{\mathrm{X}}
\newcommand{\A}{\mathbf{A}} 
\newcommand{\E}{\mathbf{E}} 
\newcommand{\M}{\mathfrak{M}}

\newcommand{\ax}{\mathbf{u}} 
\newcommand{\vcut}[1]{}

\newcommand{\Axsys}{\ensuremath{\mathsf{AxSys}}\xspace}

\newcommand{\ifff}{\leftrightarrow}

\def\lastname{Ju, Grilletti and Goranko}

\begin{document}

\begin{frontmatter}
  \title{~A logic for temporal conditionals \\ ~and a solution to the Sea Battle Puzzle}
  \vspace{-10pt}
  \author{Fengkui Ju}\footnote{Email: \href{mailto:fengkui.ju@bnu.edu.cn}{fengkui.ju@bnu.edu.cn}. Supported by the National Social Science Foundation of
China (No. 12CZX053) and the Major Program of the National Social Science Foundation of China (NO. 17ZDA026).}
  \vspace{-7pt}
  \address{School of Philosophy, Beijing Normal University} 
  \author{Gianluca Grilletti}\footnote{Email: \href{mailto:grilletti.gianluca@gmail.com}{grilletti.gianluca@gmail.com}. Supported by the European Research Council (ERC) under the European Union's Horizon 2020 research and innovation programme (No. 680220).}
  \vspace{-7pt}
  \address{Institute for Logic, Language and Computation, University of Amsterdam}
  \author{Valentin Goranko}\footnote{Email: \href{mailto:valentin.goranko@philosophy.su.se}{valentin.goranko@philosophy.su.se}. Partly supported by research grant 2015-04388 of the Swedish Research Council.}
  \vspace{-7pt}
  \address{Department of Philosophy, Stockholm University} 
  \address{\hspace{-2mm} and Department of Mathematics, University of Johannesburg (visiting professorship)}

\begin{abstract}
Temporal reasoning with conditionals is more complex than both classical temporal reasoning and reasoning with timeless conditionals, and can lead to some rather counter-intuitive conclusions. For instance, Aristotle's famous ``Sea Battle Tomorrow''  puzzle leads to a fatalistic conclusion: whether there will be a sea battle tomorrow or not, but that is necessarily the case now. We propose a branching-time logic $\mathsf{LTC}$ to formalise reasoning  about temporal conditionals and provide that logic with adequate formal semantics. The logic $\mathsf{LTC}$ extends the Nexttime fragment of $\mathsf{CTL}^*$, with operators for model updates, restricting the domain to only future moments where antecedent is still possible to satisfy. We provide formal semantics for these operators that implements the restrictor interpretation of antecedents of temporalized conditionals, by suitably restricting the domain of discourse. As a motivating example, we demonstrate that a naturally formalised in our logic version of the `Sea Battle' argument renders it unsound, thereby providing a solution to the problem with fatalist conclusion that it entails, because its underlying reasoning per cases argument no longer applies when these cases are treated not as material implications but as temporal conditionals. On the technical side, we analyze the semantics of $\mathsf{LTC}$ and provide a series of reductions of $\mathsf{LTC}$-formulae, first recursively eliminating the dynamic update operators and then the path quantifiers in such formulae. Using these reductions we obtain a sound and complete axiomatization for $\mathsf{LTC}$, and reduce its decision problem to that of the modal logic $\mathsf{KD}$.
\end{abstract}

\begin{keyword}
conditionals, temporal settings, restrictors, the Sea Battle Puzzle
\end{keyword}
\end{frontmatter}

\section{Introduction} 
\label{sec:intro}

\paragraph{Temporal conditionals}

The original philosophical motivation for the present work is rooted in the concept of \emph{conditional} (see e.g. \cite{sep-conditionals}, \cite{sep-logic-conditionals}). Classical logic only deals with the simplest kind of conditionals, viz. material implications. However, the logical laws for material implications sometimes lead to counter-intuitive, or just plainly inconsistent conclusions when applied to other types of conditionals. Here is a simple example from \cite{Veltman85}. There are two marbles in a black box and we know that one of them is blue and the other is red. Someone picks a marble from the box but we do not see which one. Intuitively the sentence ``\textit{the picked marble must be blue or the picked marble must be red}'' is false. However, each of the following sentences are true: ``\textit{the picked marble is not blue or is not red. if the picked marble is not blue then it must be red. if the picked marble is not red then it must be blue.}'' These sentences seem to imply that the picked marble must be blue or the picked marble must be red.

Some conditionals involve temporality. Here is an example from chess:
``\textit{if you attack the opponent's queen with your knight in the next move, you will eventually lose the game.}'' We call such conditionals \emph{temporal conditionals}. Various logical problems are concerned with such conditionals and a typical one is the \textit{Sea Battle Puzzle}. This puzzle goes back to Aristotle and is also closely related to the Master Argument of Diodorus Cronus. Briefly, it goes as follows\footnote{Aristotle does not present the puzzle clearly and there are different versions of it in the literature. We use the one from \cite{Garrett17}.}: either there will be a sea battle tomorrow or not; if there is a sea battle tomorrow, it is necessarily so; if there is no sea battle tomorrow, it is necessarily so. Thus, either there will necessarily be a sea battle tomorrow or there will necessarily be no sea battle tomorrow. The conclusion of this argument seems fatalistic and unacceptable for many, but its premises seem mostly fine.

\paragraph{Restrictors} Several works, including \cite{Lewis75}, \cite{vanBenthem84}, \cite{Kratzer86} and \cite{Veltman96} have proposed treatment of the semantics of conditionals based on the concept of \emph{restrictor}, which can be traced back to Ramsey's Test \cite{Ramsey29}. The restrictor-based view interprets conditionals differently from material implication. By this view, the conditional `\textsf{if $\phi$ then $\psi$}' is not a connective relating two sentences. Instead, the if-clause, `\textsf{if $\phi$}', is a device for restricting discourse domains that are collections of possibilities. Thus, `\textsf{if $\phi$ then $\psi$}' is true with respect to a domain iff $\psi$ is true with respect to the domain restricted by `\textsf{if $\phi$}'. 

Reading `\textsf{if $\phi$}' as a restrictor can resolve some logical problems involving conditionals. The two puzzles mentioned above have the same inference pattern, that is, reasoning per cases: ``\textsf{$\phi_1$ or $\phi_2$; if $\phi_1$ then $\psi_1$; if $\phi_2$ then $\psi_2$; therefore, $\psi_1$ or $\psi_2$}''. As pointed out in \cite{Cantwell08}, \cite{KolodnyMacfarlane10} and \cite{Marra18}, reasoning per cases is not generally sound under the restrictor reading of conditionals. Here is why. Assume that `\textsf{$\phi_1$ or $\phi_2$}', 
`\textsf{if $\phi_1$ then $\psi_1$}' and `\textsf{if $\phi_2$ then $\psi_2$}' are true with respect to a discourse domain $\Delta$. Let $\Delta^{\phi_1}$ and $\Delta^{\phi_2}$ be the respective results of restricting
$\Delta$ with `\textsf{if $\phi_1$}' and `\textsf{if $\phi_2$}'. What `\textsf{if $\phi_1$ then $\psi_1$}' and `\textsf{if $\phi_2$ then $\psi_2$}' say is just that $\psi_1$ is
true with respect to $\Delta^{\phi_1}$ and $\psi_2$ is true with respect to $\Delta^{\phi_2}$. But if neither the truth of $\psi_1$ nor the
truth of $\psi_2$ is upward monotonic relative to discourse domains, then
it is possible that neither $\psi_1$ nor $\psi_2$ is true with respect to $\Delta$.

Under the restrictor view of conditionals, the two puzzles are not puzzling any more. The discourse domain of the marble puzzle consists of epistemic possibilities concerning the color of the picked marble. \emph{Must} is not a upward monotonic notion relative to classes of epistemic possibilities. So the argument in this puzzle is not sound. The discourse domain of the Sea Battle Puzzle consists of possible futures and \emph{necessity} is not a upward monotonic notion relative to classes of possible futures. Then the argument in this puzzle is not sound either.

\medskip

\paragraph{Absolute and relative necessity} Temporal restrictors shrink discourse domains which consist of \emph{possible} futures. We argue that there are two senses of possibility: \emph{ontological} and \emph{genuine} possibility.

In reality we make decisions to do something or not to do something. So we will not do whatever we are able to do. Call the class of things we are able to do \emph{the ability domain}. Call the ability domain excluding the things we decide to refrain from doing \emph{the intension domain}. A future is ontologically possible if we can realize it by doing things in the ability domain. A future is genuinely possible if we can realize it by doing things in the intension domain.

Here is an example. There will be an exam tomorrow morning and there are two things for a student to do this evening: preparing for the exam and watching a football match. Doing the former will enable the student to pass the exam but doing the latter will not. The student decides to watch the match. In this case, the future where the student passes the exam is just ontologically possible but not genuinely possible.

Respectively, we can distinguish two senses of necessity: \emph{absolute} and \emph{relative} necessity. A proposition is absolutely necessary if it is the case for every ontologically possible future. A proposition is relatively necessary if it is the case for every genuinely possible future.

The second/third premise of the Sea Battle Puzzle, ``\emph{if there is a/no sea battle tomorrow, it is necessarily so}'', can be called \emph{the principle of necessity of truth}. As there are two senses of necessity, there are two principles of necessity of truth.

We argue that the principle of absolute necessity of truth does not hold. Assume that a fleet  admiral is able to do two things: $a$ and $b$. Doing $a$ will cause a sea battle tomorrow but doing $b$ will not. He decides to do $a$. In this case, it is plausible to say that there will be a sea battle tomorrow. But it is wrong to say that there will be a sea battle tomorrow no
matter what the admiral will do in the absolute sense.

We think, however, that the principle of relative necessity of truth does hold. Let $\Delta$ be the class of all the choices for which the agent is open. Assume that it is not relatively necessary that $\phi$ will be the case. Then there is a $b$ in $\Delta$ such that doing $b$ will make $\phi$ false. In this situation, it is strange to say that $\phi$ will be the case, as the agent might do $b$. So if $\phi$ will be the case, it is relatively necessary that $\phi$ will be the case.

The first premise of the Sea Battle Puzzle, ``\textit{either there will be a sea battle tomorrow or not}'', can be called \emph{the principle of excluded future middle}. We think that this principle is
commonly accepted, at least in classical (non-intuitionistic) reasoning. So all the three premises of the puzzle hold if necessity is understood in the relative sense.

However, the conclusion of this puzzle is problematic for both senses of necessity. Again, suppose that doing $a$ will cause a sea battle tomorrow but doing $b$ will not. Assume that the admiral has not made the decision to do $a$ or $b$ yet. Then it is wrong to say that there will necessarily be a sea battle tomorrow, and also wrong to say that there will necessarily be no sea battle tomorrow. So the conclusion of the puzzle is false.

Thus, the Sea Battle Puzzle is only logically problematic for the case of relative necessity, but not (necessarily) for the case of absolute necessity.

\medskip
 
\paragraph{Our technical proposal and contributions}

In this paper we introduce a new system of extended temporal logic $\mathsf{LTC}$ to formalise temporal conditionals. This logic can express relative necessity. We will show that the three premises of the Sea Battle Puzzle are valid in this logic but its conclusion is not, thus indicating that the reasoning behind it is not sound with respect to this logic. The logic $\mathsf{LTC}$ extends the Nexttime fragment of $\mathsf{CTL}^*$ with \emph{operators for model updates}, restricting the temporal domain to those future moments where antecedent is still possible to satisfy. As an illustrating example, we demonstrate that a suitably formalised in our logic version of the `Sea Battle' argument is not sound, thereby providing a solution to the problem with the fatalist conclusion that it entails. As discussed earlier, the technical core of the proposed solution is that the underlying reasoning per cases argument no longer applies when these cases are treated not as material implications but as temporal conditionals. On the technical side, we analyze the semantics of $\mathsf{LTC}$ and provide a series of reductions of $\mathsf{LTC}$-formulae, first recursively eliminating the dynamic operators and then the path quantifiers in such formulae. Using these reductions we obtain a sound and complete axiomatization for $\mathsf{LTC}$ and also eventually reduce its decision problem to that of the modal logic $\mathsf{KD}$.

Thus, we consider the contributions of this work to be two-fold: to develop a logical system for formalization and analysis of temporal conditionals, and to apply it to the clarification and solution of various philosophical logical problems that they present.

\medskip

\paragraph{Structure of the paper}

In Section \ref{section: the temporal logic with domain restrictors ltc} we introduce the logic of temporal conditionals \LTC and provide and discuss its formal semantics. In Section \ref{section: restrictors work as intended} we show that the update operator in \LTC properly captures temporal conditionals. Section \ref{section: our solution to the sea battle puzzle} proposes and discusses our  formal solution to the Sea Battle Puzzle, based on its formalisation in \LTC. In Section \ref{section: expressive power of ltc} we prove that, by means of effective translations, the model update operator can be eliminated and then the path quantifier and temporal operator in \emph{state formulas} can be replaced by the classical box modality. Using these results, in Section \ref{section: complete axiomatization of ltc} we establish a sound and complete axiomatization of \LTC and show its decidability. We conclude with brief remarks in Section \ref{section: concluding remarks}.

\section{The logic for temporal conditionals \LTC}
\label{section: the temporal logic with domain restrictors ltc}

We assume that the reader has basic familiarity with modal and branching-time temporal logics, e.g. within \cite[Chapters 4, 5, 7.1]{TLCSbook}.

Let $\Phi_0$ be a fixed countable set of atomic propositions and let $p$ range over it. The language $\Phi_{\LTC}$ of \LTC involves $\Phi_0$, the propositional connectives $\top, \neg, \land$ and the modalities $\A, \X$ and $[\cdot]$, where $\cdot$ stands for a formula (see further). $\Phi_{\LTC}$ and $\Phi_{\X [\cdot]}$, a fragment of $\Phi_{\LTC}$, are defined as follows:

\begin{displaymath}
\begin{array}{rl}
\Phi_{\X [\cdot]}: & \psi ::= p \,\,|\,\, \top \,\,|\,\, \neg \psi \,\,|\,\, (\psi \land \psi) \,\,|\,\, \X \psi \,\,|\,\, [\psi] \psi \\
\Phi_{\LTC}: & \phi ::= p \,\,|\,\, \top \,\,|\,\, \neg \phi \,\,|\,\, (\phi \land \phi) \,\,|\,\, \X \phi \,\,|\,\, \A \phi \,\,|\,\, [\psi] \phi \\
\end{array}
\end{displaymath}

\noindent The other propositional connectives: $\bot$ (falsum), $\lor, \to, \ifff$ are defined in the usual way. We also define $\mathbf{E} \phi := \neg \A \neg \phi$.

We can naturally distinguish state and path formulae of \LTC (just like in the logic $\mathsf{CTL}^*$), defined by mutual induction as follows:
\begin{displaymath}
\begin{array}{rl}
\text{State formulae } \Phi^{s}_{\LTC}: \ \ \ \ & \phi ::= p \mid \top \mid \neg \phi \mid (\phi \land \phi) \mid \A \theta \mid [\psi] \phi, \text{ where } \psi \in \Phi_{\mathsf{PC}} \\
\text{Path formulae } \Phi^{p}_{\LTC}: \ \ \ \ & \theta ::= \phi \mid \neg \theta \mid (\theta \land \theta) \mid \X \theta \mid [\psi] \theta, \text{ where } \psi \in \Phi_{\X [\cdot]} \\
\end{array}
\end{displaymath}

We define the abbreviation $\Box := \A\X$. What follow are four fragments of $\Phi_{\LTC}$ which will be used for technical purposes later.
\begin{displaymath}
\begin{array}{rl}
\Phi_{\mathsf{PC}}: & \psi ::= p \,\,|\,\, \top \,\,|\,\, \neg \psi \,\,|\,\, (\psi \land \psi) \\
\Phi_{\X}: & \psi ::= p \,\,|\,\, \top \,\,|\,\, \neg \psi \,\,|\,\, (\psi \land \psi) \,\,|\,\, \X \psi \\
\Phi_{\A\X}: & \psi ::= p \,\,|\,\, \top \,\,|\,\, \neg \psi \,\,|\,\, (\psi \land \psi) \,\,|\,\, \X \psi \,\,|\,\,\A \phi \\
\Phi_{\Box}: & \psi ::= p \,\,|\,\, \top \,\,|\,\, \neg \psi \,\,|\,\, (\psi \land \psi) \,\,|\,\, \Box \psi \\
\end{array}
\end{displaymath}

Let us give some intuition on the logical operators in \LTC. Informally, $\X \phi$ means that $\phi$ will be the case in the next moment. $\A \phi$ means that no matter how the agent will act in the future, $\phi$ is the case now, that is, $\phi$ is \emph{necessary}. $\A$ can be viewed as a universal quantifier over possible futures. Later we will see that possible futures in the semantic setting of \LTC are \emph{genuinely possible} futures in the intuitive sense. So $\A \phi$ indicates the relative necessity. Respectively, $\mathbf{E} \phi$ states that the agent has a way to act in the future so that $\phi$ is the case now, that is, $\phi$ is \emph{possible}\footnote{It seems strange to say that the agent has a way to act in the future so that $\phi$ is the case now. Actually this is fine, as whether a sentence involving future is true or not now might be dependent on how the agent will act in the future. Indeed, the truth of ``there will be a sea battle tomorrow'' depends on what the admirals of the fighting fleets will decide tonight.}.

Intuitively, $[\psi] \phi$ indicates that given $\psi$, $\phi$ is the case. Assume that one's decision to make $\psi$ true will always make $\psi$ true\footnote{In this paper we maintain the assumption that making a decision to do something will always result in this thing being done.}. Then another intended understanding of $[\psi] \phi$ is as follows: the occurrence of $\psi$ in $[\psi]$ represents the action of deciding to make $\psi$ true and $[\psi] \phi$ is read as that $\phi$ is the case after the agent decides to make $\psi$ true. This will be seen more clearly from the formal semantics.

\medskip

Now, some technical terminology and notation. Let $W$ be a nonempty set of states and let $R$ be a binary relation on it. A (finite or infinite) sequence $w_0 \dots w_n (...)$ of states is called an \emph{$R$-sequence} if $w_0 R \dots R w_n (...)$. Note that $w$ is an $R$-sequence, for any $w\in W$. Next, $(W, R)$ is a \emph{tree} if there is a state $r \in W$, called the \emph{root}, such that for any $w$, then there is a unique $R$-sequence starting with $r$ and ending with $w$. It is easy to see that if there is a root, then it is unique and $R$ is irreflexive. Further, $R$ is \emph{serial} if, for any $w\in W$ there is a $u \in W$ such that $Rwu$. We say that a tree $(W, R)$ is serial if $R$ is.

A serial tree can be understood as a time structure encoding an agent's actions (the transitions) and states in time (the nodes). A branching in the tree is interpreted as a situation where the agent can choose between different possible actions. The seriality corresponds to the fact that the agent can always perform an action at any given time.

Fix a serial tree $(W, R)$. A finite $R$-sequence $w_0 \dots w_n$ starting at the root is called a \emph{history} of $w_n$. 
For any states $w$ and $u$, $u$ is a \emph{historical} state of $w$ if there is a $R$-sequence $u_0 \dots u_n$ such that  $0 < n$, $u_0 = u$ and $u_n = w$. If $u$ is a historical state of $w$, we also say that $w$ is a \emph{future} state of $u$. Note that no state in a tree can be a historical or future state of itself.

An infinite $R$-sequence is called a \emph{path}. A path starting at the root is a \emph{timeline}. A path $\pi = w_0 w_1 \dots$ \emph{passes through} a state $x$ if $x = w_i$ for some $i$. For any path $\pi$, we use $\pi (i)$ to denote the $i+1$-th element of $\pi$, $^i \pi$ the prefix of $\pi$ to the $i+1$-th element, and $\pi^i$ the suffix of $\pi$ from the $i+1$-th element. For example, if $\pi = w_0 w_1 \dots$, then $\pi(2) = w_2$, $^2 \pi = w_0 w_1 w_2$ and $\pi^2 = w_2 w_3 \dots$. For any history $w_0 \dots w_n$ and path $u_0 u_1 \dots$, if $w_n = u_0$, let $w_0 \dots w_n \otimes u_0 \dots$ denote the timeline $w_0 \dots w_n u_1 \dots$.

$\M = (W, R, r, V)$ is a (\LTC-)\emph{model} if $(W, R)$ is a serial tree with $r$ as the root and $V$ is a valuation from $\Phi_0$ to $2^W$. Figure \ref{figure: an example of models} illustrates a model.

\begin{figure}[h]
\begin{center}
\begin{tikzpicture}[scale=0.8]

\draw[dotted] (-1,-2.5) rectangle (6,2.5);

\node[circle, draw] (w0) at (0,0) {$q$};

\node[circle, draw] (w1) at (1.5,1.5) {$p$}
edge[<-] node[auto] {} (w0);

\node[circle, draw] (w2) at (1.5,0) {$q$}
edge[<-] node[auto] {} (w0);

\node[circle, draw] (w4) at (3,0) {$q$}
edge[<-] node[auto] {} (w2);

\node[circle, draw] (w5) at (3,-1.5) {$s$}
edge[<-] node[auto] {} (w2);

\node[circle, draw] (w6) at (4.5,0) {$q$}
edge[<-] node[auto] {} (w4);

\node[circle, draw] (w7) at (4.5,1.5) {$s$}
edge[<-] node[auto] {} (w4);

\node () at (0,-0.65) {$w_0$};
\node () at (1.5,-0.65) {$w_1$};
\node () at (3,-0.65) {$w_2$};
\node () at (4.5,-0.65) {$w_3$};
\node () at (1.5,0.85) {$u$};
\node () at (4.5,0.85) {$v$};
\node () at (3,-2.15) {$x$};

\node () at (2.0,2.0) {$\udots$};
\node () at (5,2.0) {$\udots$};
\node () at (3.5,-2.0) {$\ddots$};
\node () at (5.3,0) {$\dots$};

\end{tikzpicture}
\end{center} \vspace{-5pt} \caption{A fragment of a \LTC-model.} \label{figure: an example of models}
\end{figure}

Now we will define by mutual induction two important semantic concepts: ``\emph{$\phi$ being true at a state $\pi(i)$ relative to a timeline $\pi$ in a model $\M$}'', denoted  $\M, \pi, i \Vdash \phi$ (Definition \ref{def:semantics}), and ``\emph{the result of updating $\M$ at $w$ with $\phi$}'', denoted $\M^{\phi}_w$ (Definition \ref{def:updates}). The mutual induction is needed because $\M, \pi, i \Vdash [\phi] \psi$ is defined in terms of $ \M^{\phi}_{\pi(i)}$ which, in turn, is defined in terms of $\M, \pi, i \Vdash \phi$.

\begin{definition}[Semantics]
\label{def:semantics} 
\begin{displaymath}
\begin{array}{lll}
\M, \pi, i \Vdash p & \ \Leftrightarrow \ & \pi(i) \in V(p) \\
\M, \pi, i \Vdash \top && \\
\M, \pi, i \Vdash \neg \phi & \ \Leftrightarrow \ & \text{not } \M, \pi, i \Vdash \phi \\
\M, \pi, i \Vdash \phi \land \psi & \ \Leftrightarrow \ & \M, \pi, i \Vdash \phi \text{ and } \M, \pi, i \Vdash \psi \\
\M, \pi, i \Vdash \X \phi & \ \Leftrightarrow \ & \M, \pi, i+1 \Vdash \phi \\
\M, \pi, i \Vdash \A \phi & \ \Leftrightarrow \ & \text{for any path $\rho$ starting at $\pi(i)$, } \M, {^i \pi} \otimes \rho, i \Vdash \phi \\
\M, \pi, i \Vdash [\phi] \psi & \ \Leftrightarrow \ & \M^{\phi}_{\pi(i)}, \pi, i \Vdash \psi \text{ if $(\M^{\phi}_{\pi(i)}, \pi)$ is well-defined (see below)}
\end{array}
\end{displaymath}

\vspace{-3mm}
\hspace{-3mm}
\text{(Note that $\M, \pi, i \Vdash [\phi] \psi$ holds vacuously if $(\M^{\phi}_{\pi(i)}, \pi)$ is not well-defined.)}
\end{definition}

The truth condition of $\A \phi$ at $\M, \pi, i$ can also be equivalently stated as follows: $\M, \tau, i \Vdash \phi$ for any timeline $\tau$ passing through $\pi(i)$. If $\pi$ is a timeline in $\M$ and $\pi(i) = w$, we will sometimes write $\M, \pi, w \Vdash \phi$ instead of $\M, \pi, i \Vdash \phi$. We say that $\phi$ is \emph{achievable} at $w$ in $\M$ if $\M, \pi, w \Vdash \E \phi$, i.e. $\M, \pi, w \Vdash \phi $ for some timeline $\pi$ containing $w$.

Note that the truth of a state formula at a state, relative to a timeline, is not dependent on the timeline, but this is not so for path formulae. Sometimes, if $\phi$ is a state formula, we write $\M, w \Vdash \phi$ without specifying a timeline. 

\begin{definition}[Model updates and well-definedness]
\label{def:updates}
Suppose $\phi$ is achievable at $w$ in $\M$. We define the set $X^{\phi}_w \subseteq W$ as follows: for any $x \in W$, $x \in X^{\phi}_w$ iff

\begin{enumerate}
\item $x$ is a future state of $w$, and
\item there is no timeline $\rho$ passing through $w$ and $x$ such that $\M, \rho, w \Vdash \phi$.
\end{enumerate}

Now, we define $\M^{\phi}_w = (W - X^{\phi}_w, R', r, V')$ as the \emph{restriction} of $\M$ to $W - X^{\phi}_w$, that is, $R' = R \cap (W - X^{\phi}_w)^{2}$ and $V'(p) = V(p) \cap (W - X^{\phi}_w)$ for each $p \in \Phi_0$.

If $\phi$ is not achievable at $w$, we declare $\M^{\phi}_w$ undefined because in this case all successors of $w$ are in $X^{\phi}_w$ and therefore $R'$ is not serial.

If  $\pi$ is a timeline in $\M$, such that $\M^{\phi}_{\pi(i)}$ is defined, and $\pi$ belongs to $\M^{\phi}_{\pi(i)}$, then we say that $(\M^{\phi}_{\pi(i)}, \pi)$ is \emph{well-defined}.
\end{definition}

Suppose that $\phi$ is achievable at $w$. It can be easily verified that $r \notin X^{\phi}_w$ and $X^{\phi}_w$ is closed under $R$. So $(W - X^{\phi}_w, R')$ is still a tree with $r$ as the root. It can also be verified that $R'$ is serial when 
 $\M^{\phi}_w$ is defined. Then $\M^{\phi}_w$ is a $\mathsf{LTC}$-model. Later, we will see that a timeline $\pi$ in $\M$ is a timeline in $\M^{\phi}_{\pi(i)}$ iff $\M, \pi, i \Vdash \phi$.

Updating $\M$ with $\phi$ at $w$ means to shrink $\M$ by removing the states in $X^{\phi}_w$. The set $X^{\phi}_w$ can be understood as follows. Assume that the agent is at $w$ and decides to make $\phi$ true. After the decision is made, some future states are not possible anymore. A state becomes impossible if, when the agent travels in time to it, it would no longer be possible to make $\phi$ true at $w$, no matter what happens afterwards. $X^{\phi}_w$ is the collection of these states. Figure \ref{figure: an example of updating a model with a formula} illustrates how a formula updates a model.

\begin{figure}[h]
\begin{center}
\begin{tikzpicture}[scale=0.75]


\draw[dotted] (-1,-2.5) rectangle (6,2.5);

\node[circle, draw, thick] (w0) at (0,0) {$q$};

\node[circle, draw] (w1) at (1.5,1.5) {$p$}
edge[<-] node[auto] {} (w0);

\node[circle, draw] (w2) at (1.5,0) {$q$}
edge[<-] node[auto] {} (w0);

\node[circle, draw] (w4) at (3,0) {$q$}
edge[<-] node[auto] {} (w2);

\node[circle, draw] (w5) at (3,-1.5) {$s$}
edge[<-] node[auto] {} (w2);

\node[circle, draw] (w6) at (4.5,0) {$q$}
edge[<-] node[auto] {} (w4);

\node[circle, draw] (w7) at (4.5,1.5) {$s$}
edge[<-] node[auto] {} (w4);

\node () at (0,-0.65) {$w_0$};
\node () at (1.5,-0.65) {$w_1$};
\node () at (3,-0.65) {$w_2$};
\node () at (4.5,-0.65) {$w_3$};
\node () at (1.5,0.85) {$u$};
\node () at (4.5,0.85) {$v$};
\node () at (3,-2.15) {$x$};

\node () at (2.0,2.0) {$\udots$};
\node () at (5,2.0) {$\udots$};
\node () at (3.5,-2.0) {$\ddots$};
\node () at (5.3,0) {$\dots$};

\node () at (2.5,-3) {$\M$};


\draw[dotted] (7,-2.5) rectangle (14,2.5);

\node[circle, draw, thick] (w0) at (8,0) {$q$};

\node[circle, draw] (w2) at (9.5,0) {$q$}
edge[<-] node[auto] {} (w0);

\node[circle, draw] (w4) at (11,0) {$q$}
edge[<-] node[auto] {} (w2);

\node[circle, draw] (w6) at (12.5,0) {$q$}
edge[<-] node[auto] {} (w4);

\node[circle, draw] (w7) at (12.5,1.5) {$s$}
edge[<-] node[auto] {} (w4);

\node () at (8,-0.65) {$w_0$};
\node () at (9.5,-0.65) {$w_1$};
\node () at (11,-0.65) {$w_2$};
\node () at (12.5,-0.65) {$w_3$};
\node () at (12.5,0.85) {$v$};

\node () at (13,2.0) {$\udots$};
\node () at (13.3,0) {$\dots$};

\node () at (10.5,-3) {$\M^{\X q \land \mathrm{XX} \neg s}_{w_0}$};

\end{tikzpicture}
\end{center} \vspace{-8pt} \caption{This figure illustrates how a model is updated and reduced by a formula. The model on the right is the result of updating the one on the left at $w_0$ with $\X q \land \mathrm{XX} \neg s$.} \label{figure: an example of updating a model with a formula}
\end{figure}

A formula $\phi$ of \LTC is \emph{valid} (resp., \emph{satisfiable}) if $\M, \pi, i \Vdash \phi$ for any (resp., for some) $\M$, $\pi$ and $i$. Note that a path formula $\phi$ is valid (resp., satisfiable) iff the state formula $\A\phi$ is valid (resp., $\E\phi$ is satisfiable). A set $\Gamma$ of formulae \emph{entails} a formula $\phi$, denoted $\Gamma \models \phi$, if for any $\M$, $\pi$ and $i$, if $\M, \pi, i \Vdash \psi$ for each $\psi \in \Gamma$, then $\M, \pi, i \Vdash \phi$. We write $\psi_1, \dots, \psi_n \models \phi$ for $\{\psi_1, \dots, \psi_n\} \models \phi$. As usual, we also say that $\phi$ and $\psi$ are \emph{equivalent}, denoted $\psi \equiv \phi$, if $\psi \models \phi$ and $\phi \models \psi$. Note that, as expected, $\equiv$ is a congruence with respect to all operators in the language of  $\LTC$, incl. $[\cdot]$: if $\phi, \psi \in \Phi_{\X[\cdot]}$ and $\phi \equiv \psi$, then $[\phi] \chi \equiv [\psi]\chi$ for any $\chi \in  \Phi_{\LTC}$, and if $\chi, \theta \in  \Phi_{\LTC}$ and $\chi \equiv \theta$, then $[\phi] \chi \equiv [\phi] \theta$ for any $\phi \in \Phi_{\X[\cdot]}$.

\section{Restrictors work as intended}
\label{section: restrictors work as intended}

The update with $\phi$ at a state in a model restricts the class of possible futures starting at that state. The following theorem indicates that it restricts the class of possible futures exactly as we wish: it excludes those possible futures which do not satisfy $\phi$.

For any $\phi$ in $\Phi_{\LTC}$, we define \emph{the temporal depth} $\phi^{td}$ (intuitively, $\phi^{td}$ indicates how far in the future $\phi$ can see) and \emph{the combined depth} $\phi^{cd}$ of $\phi$, as follows:

\begin{displaymath}
\begin{array}{ll}
p^{td} = 0 \ \ & \ \  p^{cd} = 0 \\
\top^{td} = 0 \ \ & \ \  \top^{cd} =  0 \\
(\neg \phi)^{td} = \phi^{td} \ \ & \ \  (\neg \phi)^{cd} = \phi^{cd} \\
(\phi \land \psi)^{td}  = \max \{\phi^{td}, \psi^{td}\} \ \ & \ \  (\phi \land \psi)^{cd}  = \max \{\phi^{cd}, \psi^{cd}\} \\
(\X \phi)^{td}  =  \phi^{td} + 1 \ \ & \ \  (\X \phi)^{cd}  =  \phi^{cd} + 1 \\
(\A \phi)^{td}  =  \phi^{td} \ \ & \ \  (\A \phi)^{cd}  =  \phi^{cd} + 1 \\
([\phi] \psi)^{td}  =  \max \{\phi^{td}, \psi^{td}\} \ \ & \ \  ([\phi] \psi)^{cd}  =  \max \{\phi^{cd}, \psi^{cd}\} \\
\end{array}
\end{displaymath}

\begin{lemma} 
\label{lemma: characterization of survival of paths}
Let $\phi$ be a formula in $\Phi_{\X [\cdot]}$, $\M$ be a model and $w$ be a state in it.

\begin{enumerate}
\item[(a)] For any natural number $n \geq \phi^{td}$ and any pair of timelines $\pi$ and $\tau$ of $\M$ passing through $w$ and sharing the same $n$ states after $w$, it holds that \\ $\M, \pi, w \Vdash \phi$ iff $\M, \tau, w \Vdash \phi$.
\item[(b)] For any timeline $\pi$ of $\M$ passing through $w$, $(\M^\phi_w, \pi)$ is well-defined iff $\M, \pi, w \Vdash \phi$.
\end{enumerate}
\end{lemma}

\begin{proof}
We will prove both claims simultaneously by induction on $\phi$. The base cases $\phi = p$ and $\phi = \top$ are straightforward. Suppose that both claims hold for all subformulae of $\phi$. We spell out here only the two non-trivial cases.

\medskip

Case $\phi = \X \psi$.

(a) Let $n \geq (\X \psi)^{td}$ and $\pi$ and $\tau$ be two timelines of $\M$ passing through $w$ and sharing the same $n$ states after $w$. Let $\pi(i) = \tau(i) = w$. We have the following equivalences:

\begin{enumerate}[leftmargin = 3em, rightmargin = 3em]
\item[] $\M, \pi, i  \Vdash \X \psi$
\item[$\Leftrightarrow$] $\M, \pi, i+1  \Vdash \psi$
\item[$^*\Leftrightarrow$] $\M, \tau, i+1  \Vdash \psi$
\item[$\Leftrightarrow$] $\M, \tau, i  \Vdash \X \psi$
\end{enumerate}

\noindent The equivalence marked by * holds by the inductive hypothesis for (a) applied to $\psi$, as $\pi$ and $\tau$ share the same $n-1$ states after $\pi(i+1)$ and $\psi^{td} \leq n-1$.

(b) Suppose $\pi$ is a timeline in  $\M$ passing through $w$ and $\pi(i) = w$. Note that $(\M^{\X \psi}_w, \pi)$ is well-defined iff $\M^{\X \psi}_w$ is defined and $\pi$ is a path in  it. Clearly if $\M, \pi, w \Vdash \X \psi$ then $\M^{\X \phi}_w$ is defined. To show that $(\M^{\mathrm{X} \psi}_w, \pi)$ is well-defined iff $\M, \pi, w \Vdash \mathrm{X} \psi$, it now suffices to show that, assuming $\M^{\X \phi}_w$ is defined, $\pi$ is a path in $\M^{\X \psi}_w$ iff $\M, \pi, w \Vdash \X \psi$. Assume $\M^{\X \psi}_w$ is defined. Then:

\begin{enumerate}[leftmargin = 3em, rightmargin = 3em]
\item[] $\pi$ is a path in $\M^{\X \psi}_w$
\item[$\Leftrightarrow$] for any $u$ of $\pi$, if $u$ is a future state of $w$, then there is a timeline $\tau$ passing through $u$ such that $\M, \tau, i \Vdash \X \psi$, i.e. $\M, \tau, i+1 \Vdash \psi$
\item[$^*\Leftrightarrow$] $\M, \pi, i+1 \Vdash \psi$
\item[$\Leftrightarrow$] $\M, \pi, i \Vdash \X \psi$
\end{enumerate}

\noindent Here is why the marked equivalence holds. The direction from the right to the left is clear. Assume the left. Let $n$ be a natural number such that $\psi^{td} \leq n$. Let $u = \pi(i + 1 + n)$. Then there is a timeline $\tau$ passing through $u$ such that $\M, \tau, i+1 \Vdash \psi$. Note that $\pi$ and $\tau$ share the same $n$ states after $\pi(i+1)$. By the inductive hypothesis for (a) applied to $\psi$, we have $\M, \pi, i+1 \Vdash \psi$.

\medskip

Case $\phi = [\psi] \chi$.

(a)  Let $n \geq ([\psi]\chi)^{td}$ and $\pi$ and $\tau$ be two timelines of $\M$ passing through $w$ and sharing the same $n$ states after $w$. We have the following equivalences:

\begin{enumerate}[leftmargin = 3em, rightmargin = 3em]
\item[] $\M, \pi, w \Vdash [\psi]\chi$
\item[$\Leftrightarrow$] $(\M^{\psi}_w, \pi)$ is not well-defined or $\M^\psi_w, \pi, w \Vdash \chi$
\item[$^*\Leftrightarrow$] $\M, \pi, w \not \Vdash \psi$ or $\M^\psi_w, \pi, w \Vdash \chi$
\item[$^*\Leftrightarrow$] $\M, \tau, w \not \Vdash \psi$ or $\M^\psi_w, \tau, w \Vdash \chi$
\item[$^*\Leftrightarrow$] $(\M^{\psi}_w, \tau)$ is not well-defined or $\M^\psi_w, \tau, w \Vdash \chi$
\item[$\Leftrightarrow$] $\M, \tau, w \Vdash [\psi] \chi$
\end{enumerate}

\noindent The first and third equivalence marked by * holds by the inductive hypothesis for (b) applied to $\psi$. The second equivalence marked by * holds by the inductive hypothesis for (a) applied to $\psi$, as $\psi^{td} \leq n$.

(b) Let $\pi$ be a timeline in $\M$ passing through $w$ and $\pi(i) = w$. Again, note that $(\M^{[\psi] \chi}_w, \pi)$ is well-defined iff $\M^{[\psi] \chi}_w$ is defined and $\pi$ is a path in it. Also note that if $\M, \pi, w \Vdash [\psi] \chi$, then $\M^{[\psi] \chi}_w$ is defined. To show that $(\M^{[\psi] \chi}_w, \pi)$ is well-defined iff $\M, \pi, w \Vdash [\psi] \chi$, it now suffices to show that, assuming $\M^{[\psi] \chi}_w$ is defined, $\pi$ is a path in  $\M^{[\psi] \chi}_w$ iff $\M, \pi, w \Vdash [\psi] \chi$. Assume that $\M^{[\psi] \chi}_w$ is defined. Then:

\begin{enumerate}[leftmargin = 3em, rightmargin = 3em]
\item[] $\pi$ is a path in  $\M^{[\psi] \chi}_w$
\item[$\Leftrightarrow$] for any $u$ of $\pi$, if $u$ is a future state of $w$, then there is a timeline $\tau$ passing through $u$ such that $\M, \tau, w \Vdash [\psi] \chi$
\item[$\Leftrightarrow$] for any $u$ of $\pi$, if $u$ is a future state of $w$, then there is a timeline $\tau$ passing through $u$ such that if $(\M^{\psi}_w, \tau)$ is well-defined, then $\M^\psi_w, \tau, w \Vdash \chi$
\item[$^*\Leftrightarrow$] for any $u$ of $\pi$, if $u$ is a future state of $w$, then there is a timeline $\tau$ passing through $u$ such that if $\M, \tau, w \Vdash \psi$, then $\M^\psi_w, \tau, w \Vdash \chi$
\item[$^*\Leftrightarrow$] if $\M, \pi, w \Vdash \psi$, then $\M^\psi_w, \pi, w \Vdash \chi$
\item[$\Leftrightarrow$] if $(\M^\psi_w, \pi)$ is well-defined, then $\M^\psi_w, \pi, w \Vdash \chi$
\item[$\Leftrightarrow$] $\M, \pi, w \Vdash \chi$
\end{enumerate}

\noindent The first equivalence marked by * holds by the inductive hypothesis for (b) applied to $\psi$. Here is why the second *-marked equivalence holds. The direction from the right to the left is clear. Assume the left. Let $n$ be a natural number such that $\psi^{td}, \chi^{td} \leq n$. Let $u = \pi(i + n)$. Then there is a timeline $\tau$ passing through $u$ such that  if $\M, \tau, w \Vdash \psi$, then $\M^\psi_w, \tau, w \Vdash \chi$. Note that $\pi$ and $\tau$ share the same $n$ elements after $w$. By the inductive hypothesis for (a) applied to $\psi$ and $\chi$, we have that if $\M, \pi, w \Vdash \psi$, then $\M^\psi_w, \pi, w \Vdash \chi$.
\end{proof}

The next corollary follows immediately from Lemma \ref{lemma: characterization of survival of paths}.

\begin{corollary} \label{corollary: the dynamic operator restricts possible futures as we wish}
Let $\M$ be a model, $w$ a state and $\phi$ a formula of $\Phi_{\X [\cdot]}$ achievable at $w$ (and therefore $\M^\phi_w$ is defined). Then any timeline $\pi$ in $\M$ passing through $w$ is also a timeline in  $\M^\phi_w$ iff $\M, \pi, w \Vdash \phi$.
\end{corollary}

The following lemma states that $\Phi_{\X [\cdot]}$ can be effectively reduced to $\Phi_{\X}$.

\begin{lemma} \label{lemma: reduction for simple consequent}
For any $\phi$ in $\Phi_{\X [\cdot]}$, there is a $\psi$ in $\Phi_{\X}$ equivalent to $\phi$.
\end{lemma}

\begin{proof}
We first show that for any $\psi \in \Phi_{\X}$, $[\phi] \psi$ is equivalent to $\phi \rightarrow \psi$. Note that the truth of the formulae in $\Phi_{\X}$ at a state relative to a path in a model is completely determined by the information of the state and the path. So for any $\psi \in \Phi_{\X}$, $\mathfrak{M}, \pi, w \Vdash \psi$ iff $\mathfrak{M}^\phi_w, \pi, w \Vdash \psi$. Let $\psi$ be in $\Phi_{\X}$. Then: $\mathfrak{M}, \pi, w \Vdash [\phi] \psi$ $\Leftrightarrow$ if $\mathfrak{M}, \pi, w \Vdash \phi$ then $\mathfrak{M}^\phi_w, \pi, w \Vdash \psi$ $\Leftrightarrow$ if $\mathfrak{M}, \pi, w \Vdash \phi$ then $\mathfrak{M}, \pi, w \Vdash \psi$ $\Leftrightarrow$ $\mathfrak{M}, \pi, w \Vdash \phi \rightarrow \psi$. For any $\chi$ in $\Phi_{\X [\cdot]}$, by use of the previous result, we can eliminate step-by-step the update operator $[\cdot]$ in $\chi$.
\end{proof}

Next theorem points out some connections between $[\phi] \psi$ and $\phi \rightarrow \psi$.

\begin{theorem} \label{theorem:collapse}
$[\phi] \psi$ collapses to $\phi \rightarrow \psi$ if either of the following holds: 
\begin{enumerate}
\item[(a)] $\phi$ is a state formula;
\item[(b)] $\psi$ is in $\Phi_{\X[\cdot]}$.
\end{enumerate}
\end{theorem}

\begin{proof}
(a) Suppose $\phi$ is a state formula. Let $\M$ be a model and $w$ a state in it. Then $\phi$ is either true or false at $w$. Suppose first that $\phi$ is true at $w$. Then the update with $\phi$ at $w$ does not change $\M$. Then for any timeline $\pi$ passing through $w$, $[\phi] \psi$ is true at $w$ relative to $\pi$ iff $\psi$ -- which is now equivalent to $\phi \rightarrow \psi$ -- is true at $w$ relative to $\pi$. Now, suppose $\phi$ is false at $w$. Then the update with $\phi$ at $w$ fails, so both $[\phi] \psi$ and $\phi \rightarrow \psi$ are trivially true at $w$ relative to any timeline.

(b) Suppose $\psi$ is in $\Phi_{\X[\cdot]}$. By Lemma \ref{lemma: reduction for simple consequent}, there is a $\psi' \in \Phi_{\X}$ equivalent to $\psi$. Then $[\phi] \psi$ is equivalent to $[\phi] \psi'$. By the proof for Lemma \ref{lemma: reduction for simple consequent}, $[\phi] \psi'$ is equivalent to $\phi \rightarrow \psi'$.
\end{proof}

Note that $[\phi] \psi$ does not generally collapse to $\phi \rightarrow \psi$ if $\phi$ is not a state formula and $\psi$ contains the operator $\mathbf{A}$. A typical example is $[\mathrm{X} s] \mathbf{A} \mathrm{X} s$: e.g., if there is a sea battle tomorrow, it is necessarily so.

\section{Our solution to the Sea Battle Puzzle}
\label{section: our solution to the sea battle puzzle}

We first show that the principle of relative necessity of truth, formalised in the logic as $[\phi] \A \phi$, is valid.

\begin{lemma} \label{lemma: having nothing to do with other timelines}
Let $(\M^{\phi}_{w}, \pi)$ be well-defined. For any $\psi$ in $\Phi_{\X [\cdot]}$, $\M^{\phi}_{w}, \pi, w \Vdash \psi$ iff $\M, \pi, w \Vdash \psi$.
\end{lemma}

\begin{proof}
By Lemma \ref{lemma: reduction for simple consequent}, $\Phi_{\X [\cdot]}$ can be reduced to $\Phi_{\X}$. As mentioned in the proof for Lemma \ref{lemma: reduction for simple consequent}, the truth of the formulae in $\Phi_{\X}$ at a state relative to a timeline in a model is completely determined by the information of the state and the timeline. Then the claim follows.
\end{proof}

\begin{theorem} \label{theorem: update works as wished}
$[\phi] \A \phi$ is valid, where $\phi \in \Phi_{\X [\cdot]}$.
\end{theorem}

\begin{proof}
Suppose $\M, \pi, w \not\Vdash [\phi] \A \phi$. This means that $(\M^{\phi}_w, \pi)$ is well-defined but
$\M^{\phi}_w, \pi, w \not\Vdash \A \phi$. The latter condition entails the existence of a timeline $\tau$ in $\M^{\phi}_w$ passing through $w$ such that $\M^{\phi}_w, \tau, w \not \Vdash \phi$. By {Lemma \ref{lemma: having nothing to do with other timelines}}, $\M, \tau, w \not \Vdash \phi$. By {Corollary \ref{corollary: the dynamic operator restricts possible futures as we wish}}, $\tau$ is not a timeline in $\M^{\phi}_w$ -- a contradiction.
\end{proof}

Let $s$ denote the proposition ``there is a sea battle''. The Sea Battle Puzzle can now be formalised as \[ \X s \lor \X \neg s, [\X s] \A \X s, [\X \neg s] \A \X \neg s \models \A \X s \lor \A \X \neg s.\] It is easy to show that $\X \phi \lor \X \neg \phi$ is valid. The validity of the other two premises follows from Theorem \ref{theorem: update works as wished}. On the other hand, it is also easy to see that $\A \X s \lor \A \X \neg s$ is not valid. Therefore, the logical argument behind the Puzzle of Sea Battle is not sound.

\begin{remark}
There have been various proposed solutions to the Sea Battle Puzzle in the literature, including 
{\L}ukasiewicz's three-valued logic \cite{McCall67}, Prior's Peircean temporal logic \cite{Prior67}, Prior's Ockhamist temporal logic \cite{Prior67}, the true futurist theory \cite{Ohrstrom81}, the supervaluationist theory \cite{Thomason70} and the relativist theory \cite{MacFarlane03}. These solutions accept the validity of the argument and argue that not all its premises are true. They focus on the following issue: how do we ascribe truth values to statements such as ``there will be a sea battle tomorrow''? These statements are called in the literature \emph{future contingents}, as they are about the future but do not have an absolute sense. In the first two solutions mentioned above, the principle of excluded future middle fails. In others, the principle of necessity of truth fails. Except for {\L}ukasiewicz's three-valued logic, the other solutions use branching time models. We refer to \cite{sep-future-contingents} for detailed comparison between these solutions.
\end{remark}

\section{Expressivity of $\mathsf{LTC}$}
\label{section: expressive power of ltc}

In this section, we show two results. Firstly, the update operator $[\phi]$ can be eliminated from any formula in $\Phi_{\LTC}$. So $\Phi_{\LTC}$ can be reduced to $\Phi_{\X\A}$. Secondly, when restricted to state formulas, $\Phi_{\X\A}$ are equally expressive with $\Phi_{\Box}$. Later we will see that the two results enable us to obtain a complete axiomatization for \LTC.

To eliminate $[\phi]$ from $[\phi] \psi$, the strategy is to `massage' $[\phi]$ into $\psi$, deeper and deeper, by applying some valid reduction principles, until it meets atomic propositions and dissolves. The only difficulty arises when $[\phi]$ meets the operator $\X$. To handle this, some preliminary treatment of $\phi$ is needed.

Recall that $\Phi_{\mathsf{PC}}$ is the set of purely propositional formulae.

\begin{lemma} \label{lemma: normal form}
For any $\phi$ in $\Phi_{\X}$, there are $\psi_1, \dots, \psi_n$ in $\Phi_{\mathsf{PC}}$ and $\chi_1, \dots, \chi_n$ in $\Phi_{\X}$ such that  $\phi$ is equivalent to $(\psi_1 \lor \X \chi_1) \land \dots \land (\psi_n \lor \X \chi_n)$.
\end{lemma}

\begin{proof}
The claim follows easily by firstly transforming $\phi$ to a CNF, where all subformulae beginning with $\X$ are treated as literals, and then applying the valid equivalences $\lnot \X \phi \equiv \X \lnot \phi$ and $\X \phi \lor \X \psi \equiv \X (\phi \lor \psi)$.
\end{proof}

\begin{lemma} \label{lemma: update with a conjunction}
For all $\phi, \psi \in \Phi_{\X [\cdot]}$ and $\chi \in \Phi_{\LTC}$, $[\phi \land \psi] \chi \leftrightarrow [\phi][\psi] \chi$ is valid.
\end{lemma}

\begin{proof}
Firstly, we show that $\M^{\phi \land \psi}_w$ is defined iff $(\M^{\phi}_w)^{\psi}_w$ is defined.

Assume that $\M^{\phi \land \psi}_w$ is defined, which implies $\M, \pi, w \Vdash \phi \land \psi$ for some timeline $\pi$ passing through $w$. Then $\M, \pi, w \Vdash \phi$ and $\M, \pi, w \Vdash \psi$. By {Lemma \ref{lemma: characterization of survival of paths}}, $\M^\phi_w$ is defined and $\pi$ is a timeline in $\M^\phi_w$. By {Lemma \ref{lemma: having nothing to do with other timelines}}, $\M^\phi_w, \pi, w \Vdash \psi$ and thus also $(\M^{\phi}_w)^{\psi}_w$ is defined.

Now, assume that $(\M^{\phi}_w)^{\psi}_w$ is defined, and thus $\M^\phi_w, \pi, w \Vdash \psi$ for some timeline $\pi$ in $\M^\phi_w$ passing through $w$. Moreover, by {Theorem \ref{theorem: update works as wished}}, $\M, \pi, w \Vdash [\phi] \A \phi$. As $(\M^\phi_w, \pi)$ is well-defined, $\M^\phi_w, \pi, w \Vdash \A \phi$, hence $\M^\phi_w, \pi, w \Vdash \phi$. Thus, $\M^\phi_w, \pi, w \Vdash \phi \land \psi$ and consequently, by {Lemma \ref{lemma: having nothing to do with other timelines}}, $\M, \pi, w \Vdash \phi \land \psi$. So, $\M^{\phi \land \psi}_w$ is defined, as wanted.

Secondly, we show $\M^{\phi \land \psi}_w = (\M^{\phi}_w)^{\psi}_w$. It suffices to show that $\M^{\phi \land \psi}_w$ and $(\M^{\phi}_w)^{\psi}_w$ have the same timelines passing through $w$. 
By Lemma \ref{lemma: characterization of survival of paths} and Lemma \ref{lemma: having nothing to do with other timelines}, the following holds:

\begin{enumerate}[leftmargin = 3em, rightmargin = 3em]
\item[] $\pi$ is a timeline in  $\M^{\phi \land \psi}_w$ through $w$
\item[$\Leftrightarrow$] $\M, \pi, w \Vdash \phi \land \psi$
\item[$\Leftrightarrow$] $\pi$ is a timeline in  $\M^{\phi}_w$ through $w$ and $\M, \pi, w \Vdash \psi$
\item[$\Leftrightarrow$] $\pi$ is a timeline in  $\M^{\phi}_w$ through $w$ and $\M^{\phi}_w, \pi, w \Vdash \psi$
\item[$\Leftrightarrow$] $\pi$ is a timeline in  $(\M^{\phi}_w)^{\psi}_w$ through $w$
\end{enumerate}

Consequently, 

\begin{enumerate}[leftmargin = 3em, rightmargin = 3em]
\item[] $\M, \pi, w \Vdash [\phi \land \psi] \chi$
\item[$\Leftrightarrow$] if $(\M^{\phi \land \psi}_w, \pi)$ is well-defined, then $\M^{\phi \land \psi}_w, \pi, w \Vdash \chi$
\item[$\Leftrightarrow$] if $\M^{\phi \land \psi}_w$ is defined and $\pi$ is a timeline in $\M^{\phi \land \psi}_w$, then $\M^{\phi \land \psi}_w, \pi, w \Vdash \chi$
\item[$\Leftrightarrow$] if $(\M^{\phi}_w)^{\psi}_w$ is defined and $\pi$ is a timeline in $(\M^{\phi}_w)^{\psi}_w$, then $(\M^{\phi}_w)^{\psi}_w, \pi, w \Vdash \chi$
\item[$\Leftrightarrow$] if $\M^{\phi}_w$ is defined, $\pi$ is a timeline in $\M^{\phi}_w$, $(\M^{\phi}_w)^{\psi}_w$ is defined and $\pi$ is a timeline in $(\M^{\phi}_w)^{\psi}_w$, then $(\M^{\phi}_w)^{\psi}_w, \pi, w \Vdash \chi$
\item[$\Leftrightarrow$] if $(\M^{\phi}_w, \pi)$ is well-defined and $((\M^{\phi}_w)^{\psi}_w, \pi)$ is well-defined, then $(\M^{\phi}_w)^{\psi}_w, \pi, w \Vdash \chi$
\item[$\Leftrightarrow$] if $(\M^{\phi}_w, \pi)$ is well-defined, then if $((\M^{\phi}_w)^{\psi}_w, \pi)$ is well-defined, then $(\M^{\phi}_w)^{\psi}_w, \pi, w \Vdash \chi$
\item[$\Leftrightarrow$] if $(\M^{\phi}_w, \pi)$ is well-defined, then $\M^{\phi}_w, \pi, w \Vdash [\psi] \chi$
\item[$\Leftrightarrow$] $\M, \pi, w \Vdash [\phi][\psi] \chi$
\end{enumerate}
\end{proof}

In the sequel, $[\phi \land \psi] \chi \leftrightarrow [\phi][\psi] \chi$ is called the axiom $\mathbf{Ax}_{[\land]}$. 

\begin{lemma} \label{lemma: update with X Phi}
Let $\phi \in \Phi_{\X [\cdot]}$ and $\M, \pi, i \Vdash \X \phi$. Then the generated submodels of $\M^{\X \phi}_{\pi(i)}$ and $\M^{\phi}_{\pi(i+1)}$ at $\pi (i+1)$ coincide.
\end{lemma}

\begin{proof}
As $\M, \pi, i \Vdash \X \phi$, $\M^{\mathrm{X} \phi}_{\pi(i)}$ is defined and $\M, \pi, i+1 \Vdash \phi$. Then $\M^\phi_{\pi(i+1)}$ is defined. By {Corollary \ref{corollary: the dynamic operator restricts possible futures as we wish}}, for every timeline $\tau$ passing through $\pi(i+1)$, we have the following equivalences: $\tau$ is a timeline in $\M^{X \phi}_{\pi(i)}$ $\Leftrightarrow$ $\M, \tau, i \Vdash \X \phi$ $\Leftrightarrow$ $\M, \tau, i+1 \Vdash \phi$ $\Leftrightarrow$ $\tau$ is a timeline in $\M^{\phi}_{\pi(i+1)}$. From this the claim follows easily. 
\end{proof}

\begin{lemma} \label{lemma: reduction axioms}
The following equivalences are valid, where $\phi \in \Phi_\mathsf{PC}$ and $\X \psi \in \Phi_{\X}$:
\begin{displaymath}
\begin{array}{rl}
\mathbf{Ax}_{[\cdot]p} : & [\phi \lor \X \psi] p \leftrightarrow ((\phi \lor \X \psi) \rightarrow p) \\
\mathbf{Ax}_{[\cdot]\top} : & [\phi \lor \X \psi] \top \leftrightarrow \top \\
\mathbf{Ax}_{[\cdot]\neg} : & [\phi \lor \X \psi] \neg \chi \leftrightarrow ((\phi \lor \X \psi) \rightarrow \neg [\phi \lor \X \psi] \chi) \\
\mathbf{Ax}_{[\cdot]\land} : & [\phi \lor \X \psi] (\chi \land \xi) \leftrightarrow ([\phi \lor \X \psi] \chi \land [\phi \lor \X \psi] \xi) \\
\mathbf{Ax}_{[\cdot]\X} : & [\phi \lor \X \psi] \X \chi \leftrightarrow ((\phi \rightarrow \X \chi) \land (\neg \phi \rightarrow \X [\psi] \chi)) \\
\mathbf{Ax}_{[\cdot]\A} : & [\phi \lor \X \psi] \A \chi \leftrightarrow ((\phi \lor \X \psi) \rightarrow \A [\phi \lor \X \psi] \chi) \\
\end{array}
\end{displaymath}
\end{lemma}

\begin{proof}
The first 4 equivalences are straightforward.

$\mathbf{Ax}_{[\cdot]\X}$. Assume $\M, \pi, i \not \Vdash [\phi \lor \X \psi] \X \chi$. Then $(\M^{\phi \lor \X \psi}_{\pi(i)}, \pi)$ is well-defined and $\M^{\phi \lor \X \psi}_{\pi(i)}, \pi, i \not \Vdash \X \chi$. By {Lemma \ref{lemma: characterization of survival of paths}}, $\M, \pi, i \Vdash \phi \lor \X \psi$. Assume $\M, \pi, i \Vdash \phi$. As $\phi$ is in $\Phi_{\mathsf{PC}}$, $\M^{\phi \lor \X \psi}_{\pi(i)} = \M$. Then $\M, \pi, i \not \Vdash \X \chi$. Then $\M, \pi, i \not \Vdash \phi \rightarrow \X \chi$. Assume $\M, \pi, i \Vdash \neg \phi$. As $\phi$ is in $\Phi_{\mathsf{PC}}$, $\M^{\phi \lor \X \psi}_{\pi(i)} = \M^{\X \psi}_{\pi(i)}$. Then $\M^{\X \psi}_{\pi(i)}, \pi, i \not \Vdash \X \chi$. Then $\M^{\X \psi}_{\pi(i)}, \pi, i+1 \not \Vdash \chi$. By {Lemma \ref{lemma: update with X Phi}}, $\M^{\psi}_{\pi(i+1)}, \pi, i+1 \not \Vdash \chi$. Then $\M, \pi, i+1 \not \Vdash [\psi] \chi$. Hence $\M, \pi, i \not \Vdash \X [\psi] \chi$. Therefore, $\M, \pi, i \not \Vdash \neg \phi \rightarrow \X [\psi] \chi$.

Assume $\M, \pi, i \not \Vdash \phi \rightarrow \X \chi$. Then $\M, \pi, i \Vdash \phi$ and $\M, \pi, i \not \Vdash \X \chi$. As $\phi$ is in $\Phi_{\mathsf{PC}}$, $\M^{\phi \lor \X \psi}_{\pi(i)} = \M$. Then $\M^{\phi \lor \X \psi}_{\pi(i)}, \pi, i \not \Vdash \X \chi$. Then $\M, \pi, i \not \Vdash [\phi \lor \X \psi] \X \chi$. Assume $\M, \pi, i \not \Vdash \neg \phi \rightarrow \X [\psi] \chi$. Then $\M, \pi, i \Vdash \neg \phi$ and $\M, \pi, i \not \Vdash \X [\psi] \chi$. As $\phi$ is in $\Phi_{\mathsf{PC}}$, $\M^{\phi \lor \X \psi}_{\pi(i)} = \M^{\X \psi}_{\pi(i)}$. Then $\M, \pi, i+1 \not \Vdash [\psi] \chi$. Then $\M^{\psi}_{\pi(i+1)}, \pi, i+1 \not \Vdash \chi$. By {Lemma \ref{lemma: update with X Phi}}, $\M^{\X \psi}_{\pi(i)}, \pi, i+1 \not \Vdash \chi$. Then $\M^{\X \psi}_{\pi(i)}, \pi, i \not \Vdash \X \chi$. Hence $\M^{\phi \lor \X \psi}_{\pi(i)}, \pi, i \not \Vdash \X \chi$. Therefore, $\M, \pi, i \not \Vdash [\phi \lor \X \psi] \X \chi$.

$\mathbf{Ax}_{[\cdot]\A}$. Let $\alpha = \phi \lor \X \psi$. Assume $\M, \pi, i \not \Vdash \alpha \rightarrow \A [\alpha] \chi$. Then $\M, \pi, i \Vdash \alpha$ but $\M, \pi, i \not \Vdash \A [\alpha] \chi$. By {Lemma \ref{lemma: characterization of survival of paths}}, $\pi$ is in $\M^\alpha_{\pi(i)}$. Then there is a timeline $\tau$ passing through $\pi(i)$ such that $\M, \tau, i \not \Vdash [\alpha] \chi$. Then $\M^\alpha_{\pi(i)}, \tau, i \not \Vdash \chi$. Then $\M^\alpha_{\pi(i)}, \tau, i \not \Vdash \A \chi$. Then $\M^\alpha_{\pi(i)}, \pi, i \not \Vdash \A \chi$. Then $\M, \pi, i \not \Vdash [\alpha] \A \chi$.

Assume $\M, \pi, i \not \Vdash [\alpha] \A \chi$. Then $(\M^\alpha_{\pi(i)}, \pi)$ is well-defined and $\M^\alpha_{\pi(i)}, \pi, i \not \Vdash \A \chi$. By {Lemma \ref{lemma: characterization of survival of paths}}, $\M, \pi, i \Vdash \alpha$. Then there is a timeline $\tau$ passing through $\pi(i)$ in $\M^\alpha_{\pi(i)}$ such that  $\M^\alpha_{\pi(i)}, \tau, i \not \Vdash \chi$. Then $\M, \tau, i \not \Vdash [\alpha] \chi$. Then $\M, \pi, i \not \Vdash \A [\alpha] \chi$. Hence $\M, \pi, i \not \Vdash \alpha \rightarrow \A [\alpha] \chi$.
\end{proof}

Recall that $\Phi{\X \A}$ denotes the language generated from $\Phi_0$ under $\X$ and $\A$.

\begin{theorem} 
\label{thm:translation} 
There is an effective translation $\mathbf{t}: \Phi_{\LTC} \to \Phi_{\X\A}$ preserving formulae up to equivalence, i.e. such that $\phi \equiv \mathbf{t}(\phi)$ for every $\phi \in \Phi_{\LTC}$.    
\end{theorem}

\begin{proof}
Define a language $\Phi_{\X \A [\lor]}$ as follows, where $\eta \in \Phi_{\mathsf{PC}}$ and $\chi \in \Phi_{\X}$:
\begin{displaymath}
\begin{array}{l}
\phi ::= p \,\,|\,\, \top \,\,|\,\, \neg \phi \,\,|\,\, (\phi \land \phi) \,\,|\,\, \X \phi \,\,|\,\, \A \phi \,\,|\,\, [\eta \lor \X \chi] \phi \\
\end{array}
\end{displaymath}

There is a translation $f: \Phi_{\LTC} \to \Phi_{\X \A [\lor]}$ preserving formulae up to equivalence. The reason is as follows. Take a formula $[\phi] \psi$ in $\Phi_{\LTC}$. By Lemma \ref{lemma: reduction for simple consequent}, $\phi$ is equivalent to some $\phi'$ in $\Phi_{\X}$. Then $[\phi] \psi$ is equivalent to $[\phi'] \psi$. By Lemma \ref{lemma: normal form}, $\phi'$ is equivalent to a conjunction $\phi'_1 \land \dots \land \phi'_n$ for some $\phi'_i = \eta_i \lor \X \chi_i$, where $\eta_i \in \Phi_{\mathsf{PC}}$ and $\chi_i \in \Phi_{\X}$, for $i =1,...,n$. Then $[\phi'] \psi$ is equivalent to $[\phi'_1\land \dots \land \phi'_n] \psi$. By Lemma \ref{lemma: update with a conjunction}, $[\phi'_1\land \dots \land \phi'_n] \psi$ is equivalent to $[\phi'_1] \dots [\phi'_n] \psi$.

Define a translation $g: \Phi_{\X \A [\lor]} \to \Phi_{\X \A}$ in the following way:

\begin{enumerate}
\item $p^g = p$
\item $\top^g = \top$
\item $(\neg \psi)^g = \neg \psi^g$
\item $(\psi \land \chi)^g = \psi^g \land \chi^g$
\item $(\X \psi)^g = \X \psi^g$
\item $(\A \psi)^g = \A \psi^g$
\item 
\begin{enumerate}
\item $([\eta \lor \X \chi] p)^g = ((\eta \lor \X \chi) \rightarrow p)^g$
\item $([\eta \lor \X \chi] \top)^g = \top^g$
\item $([\eta \lor \X \chi] \neg \psi)^g = ((\eta \lor \X \chi) \rightarrow \neg [\eta \lor \X \chi] \psi)^g$
\item $([\eta \lor \X \chi] (\psi \land \xi)^g = ([\eta \lor \X \chi] \psi \land [\eta \lor \X \chi] \xi)^g$
\item $([\eta \lor \X \chi] \X \psi)^g = ((\eta \rightarrow \X \psi) \lor (\neg \eta \rightarrow \X [\chi] \psi))^g$
\item $([\eta \lor \X \chi] \A \psi)^g = ((\eta \lor \X \chi) \rightarrow \A [\eta \lor \X \chi] \psi)^g$
\item $([\eta \lor \X \chi] [\alpha \lor \X \beta] \psi)^g = ([\eta \lor \X \chi] ([\alpha \lor \X \beta] \psi)^g)^g$
\end{enumerate}
\end{enumerate}

Now we show by three layers of induction that for any $\phi \in \Phi_{\X \A [\lor]}$, $\phi^g \in \Phi_{\X \A}$ and $\phi^g \equiv \phi$. We consider only the former result; with Lemma \ref{lemma: reduction axioms}, the argument for the latter result is similar. We put the first layer of induction on $\phi$. The only non-trivial case is $\phi = ([\eta \lor \X \chi] \gamma$. Then we put the second layer of induction on $\gamma$. We can easily go through all the subcases except the one $\gamma = [\alpha \lor \X \beta] \psi$. Assume $\gamma = [\alpha \lor \X \beta] \psi$. By the inductive hypothesis for the second layer of induction, $([\alpha \lor \X \beta] \psi)^g \in \Phi_{\X \A}$. Then by the third layer of induction on $([\alpha \lor \X \beta] \psi)^g$, we can get $([\eta \lor \X \chi] \psi)^g \in \Phi_{\X \A}$.

Now, the translation $\mathbf{t}: \Phi_{\LTC} \to \Phi_{\X \A}$ defined as $g \circ f$ satisfies the claim. 
\end{proof}

\begin{example} 
\label{example:translation}
This example demonstrates how the translation $\mathbf{t}$ is computed in the non-trivial case of $\mathbf{t} ([\phi] \psi)$. It also  illustrates the exponential blow-up that the translation causes to the length of the input formula in this case.

\begin{itemize}
\item[] $[(p \lor \X q) \land (r \lor \X s)] \X t$

\item[$\equiv$] $[p \lor \X q] [r \lor \X s] \X t$

\item[$\equiv$] $[p \lor \X q] \big( (r \rightarrow \X t) \land (\neg r \rightarrow \X [s] t) \big)$

\item[$\equiv$] $[p \lor \X q] \Big( (r \rightarrow \X t) \land \big( \neg r \rightarrow \X (s \rightarrow t) \big) \Big)$

\item[$\equiv$] $[p \lor \X q] (r \rightarrow \X t) \land [p \lor \X q] \big( \neg r \rightarrow \X (s \rightarrow t) \big)$

\item[$\equiv$] $([p \lor \X q] r \rightarrow [p \lor \X q] \X t) \land \big( [p \lor \X q] \neg r \rightarrow [p \lor \X q] \X (s \rightarrow t) \big)$

\item[$\equiv$] $\Big( \big( (p \lor \X q) \rightarrow r \big) \rightarrow [p \lor \X q] \X t \Big) \land \Big( \big( (p \lor \X q) \rightarrow \neg r \big) \rightarrow [p \lor \X q] \X (s \rightarrow t) \Big)$

\item[$\equiv$] $\bigg( \big( (p \lor \X q) \rightarrow r \big) \rightarrow \big( (p \rightarrow \X t) \land (\neg p \rightarrow \X [q] t) \big) \bigg) \land \\
\bigg( \big( (p \lor \X q) \rightarrow \neg r \big) \rightarrow \Big( \big(p \rightarrow \X (s \rightarrow t) \big) \land \big( \neg p \rightarrow \X [q] (s \rightarrow t) \big) \Big) \bigg)$

\item[$\equiv$] $\Bigg( \big( (p \lor \X q) \rightarrow r \big) \rightarrow \Big( (p \rightarrow \X t) \land (\neg p \rightarrow \X (q \rightarrow t) \big) \Big) \Bigg) \land \\
\Bigg( \big( (p \lor \X q) \rightarrow \neg r \big) \rightarrow \bigg( \big(p \rightarrow \X (s \rightarrow t) \big) \land \Big( \neg p \rightarrow \X \big( q \rightarrow (s \rightarrow t) \big) \Big) \bigg) \Bigg)$
\end{itemize}
\end{example}

\begin{lemma} \label{lem: for reduction3}
Let $\eta$ be a state formula and $\chi$ a formula in $\Phi_{\X \A}$. Then the following equivalences are valid:

\begin{enumerate}
\item $\A(\eta \lor \chi) \leftrightarrow (\eta \lor \A\chi)$
\item $\A\X\X\chi \leftrightarrow \A \X \A \X \chi$
\end{enumerate}
\end{lemma}

\begin{theorem}
\label{thm:reduction3}
There is an effective translation $\ax$ from the class of state formulas of $\Phi_{\X \A}$ to $\Phi_{\Box}$ preserving formulae up to equivalence.
\end{theorem}

\begin{proof}
(Sketch.) Induction on the combined depth $\phi^{cd}$ of $\phi$. When $\phi^{cd} = 0$ then $\phi$ is already in $\Phi_{\Box}$. Suppose $\phi^{cd} = k+1$ and the claim holds for all state formulae $\theta$ such that $\theta^{cd} \leq k$.

Now, the formula $\phi$ is a boolean combination of propositional formulae and formulae of the type $\A \psi$. In the latter, $\psi$ is equivalent to a CNF, with disjunctions of the type $\delta = \alpha \lor \E \theta \lor \A \theta_{1} \lor \ldots \lor \A \theta_{k} \lor \X\chi$,  where $\alpha \in \Phi_{\mathsf{PC}}$. Since each of $\alpha, \E \theta, \A \theta_{1},...,  \A \theta_{k}$ is a state formula of combined depth $\leq k$, by using the inductive hypothesis we can assume that each of these is already replaced by an  equivalent formula from $\Phi_{\Box}$. Then $\A \psi$ is equivalent to a conjunction of state formulae of the type $\A \delta = \A(\alpha \lor \E \theta \lor \A \theta_{1} \lor \ldots \lor \A \theta_{k} \lor \X\chi)$, which, by Lemma \ref{lem: for reduction3}, is equivalent to $\alpha \lor \E \theta \lor \A \theta_{1} \lor \ldots \lor \A \theta_{k} \lor \A\X\chi$. Note that $\chi^{cd} \leq k-1$. After transforming, likewise, $\chi$ to a CNF, distributing $\A$ over the conjunctions in it and pulling out the state subformulae from these disjunctions outside $\A$, we reduce $\A\X\chi$ above to an equivalent boolean combination of formulae from $\Phi_{\Box}$ and formulae of the type $\A\X\X\gamma$. The latter is equivalent to $\A\X\A\X\gamma$ by Lemma \ref{lem: for reduction3}. Now, since $\gamma^{cd} \leq k-2$ then $(\A\X\gamma)^{cd} \leq k$, we can use the inductive hypothesis to replace $\A\X\gamma$ with an equivalent formula from $\Phi_{\Box}$. By doing that for all such subformulae $\A\X\X\gamma$ we obtain an equivalent from $\Phi_{\Box}$ for $\A\X\chi$, and therefore can define such an equivalent $\ax(\A \delta)$ for $\A \delta$. Finally, we define $\ax(\phi)$ to be the conjunction of all these $\ax(\A \delta)$.
\end{proof}

\section{Complete axiomatization of $\LTC$}
\label{section: complete axiomatization of ltc}
\label{section:TLRaxiomatization}

Here we present a sound and complete axiomatic system $\Axsys_{\LTC}$ for the logic $\LTC$, consisting of the following groups of axioms schemes and inference rules:

\medskip
I. Any complete system of axioms for \textsf{PC}.

\medskip
II. The usual axiom schemes for $\X$:

\smallskip
\textbf{Ax$_{\mathsf{K(\X)}}$}: $\X(\phi \to \psi) \to (\X\phi \to \X\psi)$

\smallskip
\textbf{Ax$_{\mathsf{D(\X)}}$}: $\lnot \X \bot$

\smallskip
\textbf{Ax$_{\mathsf{Fun(\X)}}$}: $\lnot \X \phi \to \X \lnot \phi$

\medskip
II. The $\mathsf{S5}$ axioms for $\A$, plus the following:

\smallskip
\textbf{Ax$_{\mathsf{St(\A)}}$}: $\phi \to \A \phi$, for any state formula $\phi$\footnote{It suffices to state that axiom for atomic propositions only, and the rest would be derivable, but we prefer to streamline the deductive system.}.

\smallskip
\textbf{Ax$_{\mathsf{\A\X\X}}$}: $\A\X\X \phi \to \A\X \A\X \phi$
  
\medskip
IV. Axiom schemes for $[\cdot]$: the valid schemes $\mathbf{Ax}_{[\land]}$ from Lemma \ref{lemma: update with a conjunction} and \textbf{Ax$_{[\cdot]p}$},
\textbf{Ax$_{[\cdot]\top}$},
\textbf{Ax$_{[\cdot]\neg}$},
\textbf{Ax$_{[\cdot]\land}$},
\textbf{Ax$_{[\cdot]\X}$},
\textbf{Ax$_{[\cdot]\A}$} from Lemma \ref{lemma: reduction axioms}.

\medskip
V. Inference rules: \textsf{Modus Ponens} ($\mathsf{MP}$) and the necessitation rules:
\[
\mathsf{Nec}_{\A}:\ \frac{\vdash \phi}{\vdash \A \phi}, \ \ \ \ \ \
\mathsf{Nec}_{\X}:\ \frac{\vdash \phi}{\vdash \X \phi}, \ \ \ \ \ \ 
\mathsf{Nec}_{[\cdot]}:\ \frac{\vdash \phi}{\vdash [\psi] \phi}, \ \ \ \ \ \ 
\mathsf{Equiv}_{[\cdot]}:\ \frac{\vdash \phi \ifff \psi}{\vdash [\phi] \chi \ifff [\psi] \chi}
\]

\medskip 
The subsystem $\Axsys_{\A\X}$ consists of axioms I, II, III, and the rules in V.

\medskip 
Recall that $\Box \phi$ denotes $\A\X \phi$ and $\Phi_{\Box}$ is the respective fragment of $\Phi_{\X \A}$ built from $\Phi_{0}$ only by using the propositional connectives and $\Box$. We define the subsystem $\Axsys_{\Box}$, consisting of the $\mathsf{KD}$ axioms for $\Box$:

\smallskip
\textbf{Ax$_{\mathsf{K(\Box)}}$}: $\Box(\phi \to \psi) \to (\Box\phi \to \Box\psi)$,

\smallskip
\textbf{Ax$_{\mathsf{D(\Box)}}$}: $\lnot \Box \bot$,

plus the rules $\mathsf{MP}$ and $\mathsf{Nec}_{\Box}: \ \ \frac{\vdash \phi}{\vdash \Box \phi}$.

\smallskip
Note that:

\begin{enumerate}
\item The axioms and rules of $\Axsys_{\Box}$ are derivable in  $\Axsys_{\A\X}$.
\item $\Axsys_{\Box}$, being canonical, is sound and complete for the validities in $\Phi_{\Box}$.
\end{enumerate}

Now, we are going to establish a series of equivalent reductions of formulae of \LTC, all derivable in $\Axsys_{\LTC}$, which will enable us to eventually reduce the proof of completeness of $\Axsys_{\LTC}$ to the completeness of the much simpler logic $\Axsys_{\Box}$. We begin with some useful derivations, used further in the proofs.

\begin{lemma}
\label{lem:two derivations}
The following are derivable in $\Axsys_{\LTC}$:

\begin{enumerate}
\item $\vdash_\mathsf{LTC} [\bot] \phi \leftrightarrow \top$
\item $\vdash_\mathsf{LTC} [\phi] \psi \leftrightarrow (\phi \rightarrow \psi)$, where $\phi$ is in $\Phi_{\mathsf{PC}}$
\end{enumerate}
\end{lemma}

\begin{proof}
(i) By induction on $\phi$. The only non-trivial case is $\phi = \X \psi$. Since $\vdash_\mathsf{LTC} \bot \leftrightarrow (\bot \lor \X \bot)$, we obtain $\vdash_\mathsf{LTC} [\bot] \X \psi \leftrightarrow [\bot \lor \X \bot] \X \psi$ by $\mathsf{Equiv}_{[\cdot]}$. Then, using $\mathbf{Ax}_{[\cdot]\X}$ we derive $\vdash_\mathsf{LTC} [\bot \lor \X \bot] \X \psi \leftrightarrow ((\bot \rightarrow \X \psi) \land (\neg \bot \rightarrow \X [\bot] \psi))$. By the inductive hypothesis, $\vdash_\mathsf{LTC} [\bot] \psi \leftrightarrow \top$. Then $\vdash_\mathsf{LTC} [\bot \lor \X \bot] \X \psi \leftrightarrow ((\bot \rightarrow \X \psi) \land (\neg \bot \rightarrow \X \top))$. Therefore, $\vdash_\mathsf{LTC} [\bot \lor \X \bot] \X \psi \leftrightarrow ((\bot \rightarrow \X \psi) \land (\neg \bot \rightarrow \top))$. Hence,  $\vdash_\mathsf{LTC} [\bot \lor \X \bot] \X \psi \leftrightarrow \top$. Therefore, $\vdash_\mathsf{LTC} [\bot] \X \psi \leftrightarrow \top$.

\medskip

(ii) By induction on $\psi$. The only non-trivial cases are $\psi = \X \chi$ and $\psi = \mathbf{A} \chi$.

Case $\psi = \X \chi$. As $\vdash_\mathsf{LTC} \phi \leftrightarrow (\phi \lor \X \bot)$, we derive $\vdash_\mathsf{LTC} [\phi] \X \chi \leftrightarrow [\phi \lor \X \bot] \X \chi$ by $\mathsf{Equiv}_{[\cdot]}$. Then, $\vdash_\mathsf{LTC} [\phi \lor \X \bot] \X \chi\leftrightarrow ((\phi \rightarrow \X \chi) \land (\neg \phi \rightarrow \X [\bot] \chi))$, by $\mathbf{Ax}_{[\cdot]\X}$. By claim (i), $\vdash_\mathsf{LTC} [\bot] \phi \leftrightarrow \top$. Hence, $\vdash_\mathsf{LTC} [\phi \lor \X \bot] \X \chi\leftrightarrow ((\phi \rightarrow \X \chi) \land (\neg \phi \rightarrow \X \top))$. Then $\vdash_\mathsf{LTC} [\phi \lor \X \bot] \X \chi\leftrightarrow ((\phi \rightarrow \X \chi) \land (\neg \phi \rightarrow \top))$, so $\vdash_\mathsf{LTC} [\phi \lor \X \bot] \X \chi\leftrightarrow (\phi \rightarrow \X \chi)$.

Case $\psi = \mathbf{A} \chi$. By $\mathbf{Ax}_{[\cdot] \A}$ we have $\vdash_\mathsf{LTC} [\phi] \A \chi \leftrightarrow (\phi \rightarrow \A [\phi] \chi)$. By the inductive hypothesis, $\vdash_\mathsf{LTC} [\phi] \chi \leftrightarrow (\phi \rightarrow \chi)$. Then $\vdash_\mathsf{LTC} [\phi] \A \chi \leftrightarrow (\phi \rightarrow \A (\phi \rightarrow \chi))$. Note that $\phi \in \Phi_{\mathsf{PC}}$. Then $\vdash_\mathsf{LTC} \A (\phi \rightarrow \chi) \leftrightarrow (\phi \rightarrow \A \chi)$, hence $\vdash_\mathsf{LTC} [\phi] \A \chi \leftrightarrow (\phi \rightarrow \A \chi)$.
\end{proof}

\begin{lemma}
\label{lem:reduction1}
For every \LTC-formula $\phi$, $\vdash_{\Axsys_{\LTC}} \mathbf{t}(\phi) \leftrightarrow \phi$.
\end{lemma}

\begin{proof}
By induction on $\phi$, by applying the axioms for $[\cdot]$, the derivations in Lemma \ref{lem:two derivations}, and the inference rules of $\Axsys_{\LTC}$, one can formalise the proof of Theorem \ref{thm:translation} in $\Axsys_{\LTC}$. The details are routine and we leave them out.
\end{proof}

\begin{lemma}
\label{lem:reduction2}
For every \LTC-formula $\phi$, if $\phi$ is $\Axsys_{\LTC}$-consistent then $\mathbf{t}(\phi)$ is $\Axsys_{\A\X}$-consistent.
\end{lemma}

\begin{proof}
Equivalently (since $\mathbf{t}(\lnot \phi) =  \lnot \mathbf{t}(\phi)$), we have to prove that for every \LTC-formula $\phi$, if $\vdash_{\Axsys_{\A\X}} \mathbf{t}(\phi)$ then $\vdash_{\Axsys_{\LTC}} \phi$. Indeed, take a derivation $\vdash_{\Axsys_{\A\X}} \mathbf{t}(\phi)$. It is also a derivation in $\Axsys_{\LTC}$. Append to it the derivation $\vdash_{\Axsys_{\LTC}} \mathbf{t}(\phi) \leftrightarrow \phi$ from Lemma \ref{lem:reduction1}, to obtain $\vdash_{\Axsys_{\LTC}} \mathbf{t}(\phi) \rightarrow \phi$. Then, by $\mathsf{MP}$ we obtain $\vdash_{\Axsys_{\LTC}} \phi$.
\end{proof}

\begin{lemma} \label{lem: ax and box}
Every state formula $\psi \in \Phi_{\X \A}$ is provably equivalent in $\Axsys_{\A\X}$ to $\ax(\psi) \in \Phi_{\Box}$.
\end{lemma}

\begin{proof} 
The equivalences used in defining the function $\ax$ in the proof for Theorem \ref{thm:reduction3},  including those in Lemma \ref{lem: for reduction3}, are all derivable in $\Axsys_{\A\X}$. 
\end{proof}

\begin{lemma}
\label{lem:reduction4}
For every formula $\psi \in \Phi_{\X \A}$ there is an effectively computable formula $\mathbf{v} (\psi) \in \Phi_{\Box}$ such that $\psi$ and $\mathbf{v} (\psi)$ are equally satisfiable and equally consistent in $\Axsys_{\A\X}$.
\end{lemma}

\begin{proof}
(Sketch.) Here `consistency' will mean `consistency in $\Axsys_{\A\X}$'. We will prove the claim for consistency; the proof  for satisfiability is fully analogous.

The definition of $\mathbf{v} (\psi)$ extends that of $\ax(\phi)$ for state formulae $\phi$ in Lemma \ref{lem: ax and box} as follows. The formula $\psi$ is provably in $\Axsys_{\A\X}$ equivalent to some, hereafter fixed, formula $\psi'$ in DNF, where every disjunct is of the type $\delta = \theta \land \X\chi$, for some state formula $\theta$ (note that $\chi$ may also be $\top$). Then, consistency of $\psi$ is equivalent (provably in $\Axsys_{\A\X}$) to consistency of some of these disjuncts. Note that $\theta \land \X\chi$ is consistent iff $\theta \land \E\X\chi$ is so. Indeed, $\vdash_{\Axsys_{\A\X}}  \theta  \to \lnot \X\chi$ iff $\vdash_{\Axsys_{\A\X}} \A(\theta  \to \lnot \X \chi)$ iff $\vdash_{\Axsys_{\A\X}} \theta  \to \A \lnot \X \chi$. Lastly, $\theta \land \E\X\chi$ is consistent iff $\ax(\theta) \land \ax(\E\X\chi)$ is so. Now, if $\psi$ is consistent we define $\mathbf{v} (\psi)$ to be $\ax(\theta) \land \ax(\E\X\chi)$ for the first disjunct $\delta$ in $\psi'$ that is consistent, else we define $\mathbf{v} (\psi)$ to be $\bot$. The claim of the lemma follows by construction.
\end{proof}

\begin{theorem}
\label{thm:completeness}
$\Axsys_{\LTC}$ is sound and complete for the logic \LTC.
\end{theorem}

\begin{proof}
The soundness of the axioms for $[\cdot]$ follows from Lemmas \ref{lemma: update with a conjunction}, 
and \ref{lemma: reduction axioms}. The soundness of Ax$_{\mathsf{\A\X\X}}$ is by Lemma \ref{lem: for reduction3}. All others axioms are well-known and proving their soundness, as well as that of the inference rules, is routine. 

\medskip

For the completeness, we need to show that every $\Axsys_{\LTC}$-consistent formula $\phi$ is satisfiable in a \LTC-model. We prove that by a series of reductions, as follows. Let $\phi$ $\Axsys_{\LTC}$-consistent. Then, by Lemma \ref{lem:reduction2}, $\mathbf{t}(\phi)$ is $\Axsys_{\A\X}$-consistent.  Next, by Lemma \ref{lem:reduction4}, $\mathbf{v} (\mathbf{t}(\phi))$ is $\Axsys_{\A\X}$-consistent, and therefore $\Axsys_{\Box}$-consistent. Since $\Axsys_{\Box}$ is complete for the logic $\textsf{KD}$, it follows that $\mathbf{v} (\mathbf{t}(\phi))$ is satisfiable, hence $\mathbf{t}(\phi)$ is satisfiable by Lemma \ref{lem:reduction4}, and therefore $\phi$ is satisfiable, too, by Theorem \ref{thm:translation}. 
\end{proof}

\begin{theorem}
\label{thm:decidability}
The satisfiability problem for \LTC is decidable, in \textsc{ExpSpace}.
\end{theorem}

\begin{proof} The translations and reductions defined here effectively reduce the satisfiability problem for \LTC to that of $\textsf{KD}$, which is \textsc{PSpace}-complete (see e.g. \cite{MARX2007139}). However, as illustrated by Example \ref{example:translation}, the translation $\mathbf{t}$ can cause an exponential blow-up to the length of the formula, hence the complexity upper bound that we provide based on the results obtained here is  \textsc{ExpSpace}. However, the number of different subformulae need not grow exponentially, so it is conceivable that a more refined, on-the-fly application of the translation $\mathbf{t}$, or an alternative decision method may bring that complexity down. We leave the question of establishing the precise complexity of \LTC open.
\end{proof}

\section{Concluding remarks} 
\label{section: concluding remarks}

In this work we have proposed the logic \LTC for formalising reasoning with and about temporal conditionals, for which we establish a series of technical results, eventually leading to a sound and complete axiomatization and a decision procedure. We show how formalization on \LTC can provide a solution to the Sea Battle Puzzle, by showing that the underlying reasoning per cases argument is not sound in that logic.

The present work suggests a variety of interesting follow-up developments:
\begin{itemize}
\item On the conceptual side, we believe that the logic \LTC or suitable variations of it can be applied to resolve other philosophical problems arising from reasoning about  temporal conditionals. We also intend to adapt that logic for reasoning about temporal counterfactuals, as well as to related type of reasoning, that often occurs in the theory of extensive forms, about the players' behaviour on and off the equilibrium path, non-credible threats, etc., leading, inter alia, to the more refined notion of subgame-perfect equilibrium and myopic equilibrium. 

\item On the technical side, we note that the Nexttime fragment of the logic CTL*, which is precisely the fragment $\Phi_{\X \A}$ of \LTC, seems not to have been studied so far. The present work provides, inter alia, a new complete axiomatization result for it. We intend to extend the present results to extensions of \LTC, also involving long-term temporal operators, by relating them with more expressive, yet still manageable fragments of CTL*. 

\item In particular, an extension of the logic LTC with past and future temporal modalities can likewise resolve the problem raised by the Diodorean Master Argument, by rendering that argument unsound with respect to the resulting semantics.  
\end{itemize}

\vspace{10pt}

\subsection*{Acknowledgment} Thanks go to Maria Aloni, Johan van Benthem, Alessandra Marra, Floris Roelofsen, Frank Veltman, and the audience of seminars or workshops at Delft University of Technology, Beijing Normal University and Tsinghua University, for their useful comments and suggestions. 
Valentin Goranko is grateful to the Department of Philosophy of Beijing Normal University for supporting his visit there, during which part of the present work was done.  
Thanks are also due to the three anonymous referees for helpful comments and some corrections. 

\bibliographystyle{aiml18}
\bibliography{Sea}

\end{document}